\documentclass[graybox, envcountchap]{svmult}

\usepackage{imakeidx}         
\usepackage{graphicx}        
                             
                             \usepackage{subfigure}
                             \usepackage{algorithm,algorithmic}
                             \usepackage{wrapfig}
                             \usepackage{bm}
                             \usepackage{multirow}
                             
\usepackage{multicol}        
\usepackage[bottom]{footmisc}

\usepackage{newtxtext}       %
\usepackage{newtxmath}       

\makeindex          

\begin{document}
\mainmatter
\newcommand{\ket}[1]{{\left\vert{#1}\right\rangle}}

\title*{Quantum Algorithms for String Processing}
\author{Farid Ablayev \and Marat Ablayev \and Kamil Khadiev \and Nailya Salihova \and Alexander Vasiliev}
\index{Ablayev, Farid}
\index{Ablayev, Marat}
\index{Khadiev, Kamil}
\index{Salihova, Nailya}
\index{Vasiliev, Alexander}
\authorrunning{F. Ablayev, M. Ablayev, K. Khadiev, N. Salihova, A. Vasiliev}
\institute{F. Ablayev, M. Ablayev, K. Khadiev,  A. Vasiliev \at Kazan Federal University, Kazan, Russian Federation; Kazan E. K. Zavoisky Physical-Technical Institute, Kazan, Russian Federation,
\email{fablayev@gmail.com}
\and
N. Salihova \at Kazan Federal University, Kazan, Russian Federation
}
%
%
\maketitle

\abstract{
In the paper, we investigate two problems on strings. The first one is the String matching problem, and the second one is the String comparing problem. We provide a quantum algorithm for the String matching problem that uses exponentially less quantum memory than existing ones. The algorithm uses the hashing technique for string matching, quantum parallelism, and ideas of Grover's search algorithm. Using the same ideas, we provide two algorithms for the String comparing problem. These algorithms also use exponentially less quantum memory than existing ones. Additionally, the second algorithm works exponentially faster than the existing one.
}

\section{Introduction}
Possibilities of quantum speedup for string matching problem have been investigated during the last decades by different authors \cite{ramesh2003string, montanaro2017quantum, soni2020pattern}. Most of these algorithms are based on Grover's algorithm \cite{grover1996fast,bbht98} for search through unstructured data.

In the paper we consider a problem of searching any occurrence of a string $w$ of length $m$ in a string $s$ of length $n$. The best known classical algorithm for this problem is Knuth-Morris-Pratt algorithm \cite{knuth1977fast}. Time complexity of this algorithm is $O(n+m)$. Quantum algorithms for this problem are typically considered in the query model \cite{nc2010,a2017,aazksw2019part1}. Here the algorithm has access to an oracle (the unchangeable part of memory that holds input data) and complexity is a number of queries to this oracle. In the early 2000s researchers obtained one of the first results on quantum algorithms for the problem \cite{ramesh2003string}. This algorithm has query complexity $O(\sqrt{n}\log{\sqrt{\frac{n}{m}}}\log{m} + \sqrt{m}\log^2{m})$. Later in 2017 the algorithm \cite{montanaro2017quantum} with query complexity  $\Tilde{O}(\sqrt{\frac{n}{m}}2^{O(\sqrt{\log{m}})})$ was presented. In 2020, Soni and Rasool \cite{soni2020pattern} suggested an algorithm with $O(n \log n)$ query complexity. Note that, these algorithms have time complexity (or circuit complexity) logarithmic times larger than query complexity.

At the same time, researchers did not pay attention to the size of the quantum memory that we should use for algorithms (a changeable by algorithm part and the unchangeable part that oracle holds). Due to our analysis, all of these algorithms use $O(n+m)$ quantum bits (including the unchangeable part of the quantum memory). Due to the restricted resources of the current and near-future devices, the quantum memory size, even unchangeable, is an important question.

In the paper, we provide a quantum algorithm for solving string matching problem with $O(\sqrt{n}(\log n + \sqrt{\log n+\log m}\cdot (\log\log n+\log\log m)))$ time complexity and $O((\log n)^2+\log n \cdot \log m)$ qubits of memory (Theorem \ref{th:sm}). The algorithm is based on Grover's search algorithm, the idea of hashing (or fingerprinting method \cite{Fre79}) and ideas of Rabin-Karp algorithm \cite{kr87}.
The algorithm is not a query model algorithm but a quantum circuit algorithm that can be used as a part of other more complex algorithms for other problems. Many known algorithms like \cite{hhl2009} use similar motivation.
Our algorithm assumes that the initial state is prepared. At the same time, this initial state can be prepared approximately as fast as loading data to unchangeable memory for oracle. 

Additionally, we use the same ideas for comparing two strings $u$ and $v$ in lexicographical order. The existing algorithm \cite{ki2019} uses modifications \cite{k2014,ll2015,ll2016} of Grover's search \cite{grover1996fast} and compares two strings with query complexity $O(\sqrt{k})$, time complexity $O(\sqrt{k}\log k)$ and uses $O(k)$ qubits of memory, where $k$ is the minimum of lengths of two strings. Here we use the idea with hashing and provide two algorithms. The first one has $O(\sqrt{k}\log k)$ time complexity and uses $O((\log k)^2)$ qubits (Theorem \ref{th:comp1}). The second one has $O((\log k)^{2}\log \log k)$ time complexity and uses $O((\log k)^2)$ qubits (Theorem \ref{th:comp2}). Both algorithms have an exponential advantage in memory and the second one has an exponential advantage in speed. At the same time, the second algorithm is more complex.

The structure of the paper is the following. Section \ref{sec:prelims} contains preliminaries. We present an algorithm for string matching in Section \ref{sec:smatch}. Section \ref{sec:comp} contains algorithms for comparing two strings. The conclusion is presented in Section \ref{sec:conclusion}.

\section{Preliminaries}\label{sec:prelims}
Let us consider a string $u=(u_1,\dots,u_\ell)$ for some integer $\ell$. Then, $|u|=\ell$ is the length of the string. $u[i,j]=(u_i,\dots,u_j)$ is a substring of $u$. 

In the paper, we compare strings in the lexicographical order. For two strings $u$ and $v$, the notation $u<v$ means $u$ precedes $v$ in the lexicographical order.

In the paper, we consider only binary strings. At the same time, all results can be easily modified for a non-binary alphabet.
\subsection{Rolling Hash for Strings Comparing}\label{sec:roll-hash}
\subsubsection{Rolling Hash}
The rolling hash was presented in \cite{Fre79,kr87}.
For a string $u=(u_1,\dots,u_{|u|})$ we define a rolling hash function 
$h_p(u)=\left(\sum_{i=1}^{|u|}u_i\cdot 2^{i-1}\right)\mbox{ mod }p,$
where $p$ is a prime integer. 

\subsubsection{Fingerprinting Technique for Comparing Strings}
We can use the rolling hash and the fingerprinting method \cite{Fre79} for comparing two strings $u$ and $v$. Let us randomly choose $p$ from the set of the first $r$ primes, such that $r\leq \frac{\max(|u|,|v|)}{\varepsilon}$ for some $\varepsilon>0$. According to the Chinese Remainder Theorem and \cite{Fre79}, if we have $h_p(u)=h_p(v)$, then $u=v$ with error probability at most $\varepsilon$. If we invoke a comparing procedure $\delta$ times, then we should choose a prime number from the first $\frac{\delta \cdot \max(|u|,|v|)}{\varepsilon}$ primes for getting the error probability $\varepsilon$ for the whole algorithm. Due to Chebyshev's theorem, the $r$-th prime number $p_r\approx r\ln r$. If $r=\frac{\delta \cdot \max(|u|,|v|)}{\varepsilon}$, then $p_r=\frac{\delta \cdot \max(|u|,|v|)}{\varepsilon}\cdot (\ln(\delta) + \ln(\max(|u|,|v|))-\ln(\varepsilon))$ and it can be encoded using $O(\log(\delta) + \log(\max(|u|,|v|))-\log(\varepsilon))$ bits.

\subsubsection{Comparing Strings Using a Rolling Hash} \label{sec:roll-hash-bins}
For a string $u$, we can compute a prefix rolling hash, that is
$h_p(u[1,i])$. It can be computed in $O(|u|)$ running time using formula
\[h_p(u[1,i])=\left(h_p(u[1,i-1])+(2^{i-1}\mbox{ mod }p)\cdot u_i\right)\mbox{ mod }p\mbox{ and }h_p(u[1:0])=0.\]

Assume, that we have computed prefix rolling hashes for two strings $u$ and $v$. Then, we can compare these strings in the lexicographical order in $O(\log \min(|u|,|v|))$ running time. The algorithm is following. We search the longest common prefix of $u$ and $v$. Let $lcp(u,v)$ be an integer $x$ such that $u_1=v_1,\dots,u_x=v_x$ and $u_{x+1}\neq v_{x+1}$. In the case of $u$ is a prefix of $v$, then $lcp(u,v)=|u|$. In the case of $v$ is a prefix of $u$, we have $lcp(u,v)=|v|$. Notice, that for any integer $mid\in\{1,\dots, \min(|u|,|v|)\}$ the following two statements are true.
\begin{itemize}
    \item If $mid\leq lcp(u,v)$, then $u[1,mid]=v[1,mid]$, and $h_p(u[1,mid])=h_p(v[1,mid])$.
    \item If $mid> lcp(u,v)$, then $u[1,mid]\neq v[1,mid]$, and $h_p(u[1,mid])\neq h_p(v[1,mid])$ with high probability.
\end{itemize}

Using binary search we find the index $x$ such that $h_p(u[1,x])=h_p(v[1,x])$ and $h_p(u[1,x+1])\neq h_p(v[1,x+1])$. In that case $lcp(u,v)=x$.
  After that, we compare $u_{t}$ and $v_t$ for $t=lcp(u,v)+1$. Then, we get the following cases:
  \begin{itemize}
      \item If $u_t<v_t$ or $t=|u|<|v|$, then $u<v$.
      \item If $u_t>v_t$ or $t=|v|<|u|$, then $u>v$.
      \item If $|u|=|v|=t-1$, then $u=v$.
  \end{itemize}
Binary search works in $O(\log (\min(|u|,|v|)))$ running time.

\subsection{Problems}
\paragraph{String Matching Problem}
Given a string (text) $s=(s_1,\dots,s_n)$ of length $n$ and a string $w$ of length $m$,  where $m\leq n$, one needs to determine the index of the string  $w$ occurrence in the text $s$. Formally, the task is to find the index $d$ such that $w=(s_{d}\dots s_{d+m-1})$.

We use the following notations. Let $T(s)=(s^1, \dots, s^{n-m+1})$, where $s^i=s[i,i+m-1]$ for $i\in\{1,\dots,n-m+1\}$. $T(s)$ is a sequence of substrings of length $m$. Let $N=n-m+1$. 

\paragraph{String Comparing Problem}
Given two strings $u$ and $v$. The problem is comparing these two strings in lexicographical order. Formally, we want to determine one of three options:
\begin{itemize}
      \item If $u<v$, then the result is $-1$.
      \item If $u>v$, then the result is $+1$.
      \item If $u=v$, then the result is $0$.
  \end{itemize}

\subsection{Basics of Quantum Computation and Computational Model}
The main difference between quantum computation and the classical one is manipulations with quantum bits (qubits). A state of a qubit is a vector from $2$-dimensional complex Hilbert space. We can represent it using Dirac notation as $|\psi\rangle =a|0\rangle+b|1\rangle$, where $|0\rangle$ and $|1\rangle$ are unit vectors, and $a$ and $b$ are complex numbers such that $|a|^2+|b|^2=1$. We can use two kinds of transformations: \emph{transition} and \emph{measurement}. The transition is multiplying a vector of state to $2\times 2$ unitary matrix. The measurement is obtaining $0$-result with probability $|a|^2$ and $1$-result with probability $|b|^2$. Similarly, a state of a register of $q$ qubits is a vector from $2^q$-dimensional complex Hilbert space, and is traditionally denoted as $|\psi\rangle=\sum_{i=0}^{2^q-1}a_i|i\rangle$, where $\sum_{i=0}^{2^q-1}|a_i|^2=1$. 
Transformations are defined in an analogous manner.

A quantum circuit is a circuit that uses four types of gates that are $1$-qubit Hadamard gate ($H$-gate), $X$-gate and $Z$-gate; and $2$-qubit $CNOT$-gate. That are
\[\begin{array}{llll}
X=\left(\begin{array}{cc}0 & 1\\1 & 0\end{array}\right),&\quad Z =\left(\begin{array}{cc}1 & 0\\0 & -1\end{array}\right),&\quad H=\frac{1}{\sqrt{2}}\left(\begin{array}{cc}1 & 1\\1 & -1\end{array}\right),&\quad CNOT=\left(\begin{array}{cccc}1 & 0&0&0\\0 & 1&0&0\\0 & 0&0&1\\0 & 0&1&0\end{array}\right).
\end{array} 
\]
An algorithm's time complexity is the size of a circuit that uses only presented gates and implements the algorithm.

The standard form of the quantum query model is a generalization of the decision tree model of classical computation that is commonly used to lower-bound the amount of time required by a computation.
    Let $f:D\rightarrow \{0,1\},D\subseteq \{0,1\}^M$ be an $M$ variable function we wish to compute on an input $x=(x_0,\dots,x_{M-1})\in D$. We have an oracle access to the input. That is implemented by storing the input into an unchangeable part of quantum memory $\ket{x}$. The oracle access is realized by a specific unitary transformation usually defined as
    {%
    \relpenalty 100000
    \binoppenalty 100000
    $\ket{i}\ket{\phi}\ket{\psi}\ket{x}\rightarrow \ket{i}\ket{\phi \oplus x_i}\ket{\psi}\ket{x}$
    } where the $\ket{i}$ register indicates the index of the variable we are querying, $\ket{\phi}$ is the output register, and $\ket{\psi}$ is some auxiliary work-space. An algorithm in the query model consists of alternating applications of arbitrary unitaries (that are independent of the input) and the input-dependent query unitary, and a measurement in the end. The smallest number of queries for an algorithm that outputs $f(x)$ with probability $\geq \frac{2}{3}$ on all $x$ is called the quantum query complexity of the function $f$.
    
More information on quantum computation and computational models can be found in \cite{nc2010,a2017,aazksw2019part1}.

\section{Quantum Algorithm for String Matching Problem}\label{sec:smatch}
Firstly, let us present Grover's search algorithm because we use its ideas as a base for our algorithm.
\subsection{Grover's Search Algorithm}\label{sec:grover}

\begin{definition}[Search problem]
Suppose we have a set of objects named $\{1,2,\dots, M\}$, of which some are targets. Suppose  $\mathcal{O}$ is an oracle that identifies the targets.
The goal of a search problem is to find a target $i \in \{1,2,\dots, M\}$ by making queries to the oracle $\mathcal{O}$.
\end{definition}
Remind that Oracle is implemented by accessing an unchangeable (by the algorithm) part of the quantum memory.

In search problems, one will try to minimize the number of queries to the oracle. In the classical setting, one needs $O(M)$ queries to solve such a problem. Grover, on the other hand, constructed a quantum algorithm that solves the search problem with only $O(\sqrt{M})$ queries~\cite{grover1996fast}, provided that there is a unique target. 

The algorithm uses additional $\log M$ qubits for indexing element in a state $\frac{1}{\sqrt{M}}\sum_{t=0}^{M-1}\ket{t}$ and one additional qubit $\ket{\xi}$ in a state $\frac{1}{\sqrt{2}}(\ket{0}-\ket{1})$. On step of the algorithm is applying two operations: Grover's diffusion $D$ and a query to oracle $Q$. The matrix $D$ can be implemented using $\log M$ gates due to \cite{grover1996fast}. 

The matrix $Q$ is a transformation that converts $\ket{t}\ket{\xi\oplus f(t)}=(-1)^{f(t)}\ket{t}\ket{\xi}$, where $f(t)$ is a Boolean function that shows whether $t$-th object is target.

After $O(\sqrt{M})$ iterations, the algorithm measures the quantum register and obtains the index of the target object with high probability.
If there are no target objects, then the algorithm returns any object with equal probability.

When the number of targets is unknown, Brassard \emph{et al.} designed a modified Grover algorithm that solves the search problem with $O(\sqrt{M})$ queries~\cite{bbht98}, which is of the same order as the query complexity of the Grover search.

The algorithm repeats Grover's search algorithm for $\log_2 (\sqrt{M})$ times. It does $2^j$ iterations on $j$-th repetition. Such behavior allows us to obtain one of the target objects with a probability at least $1/2$.
\subsection{Our Algorithm}

Let us choose a prime $p$ from the first $\frac{\delta \cdot m}{\varepsilon}$ prime numbers, where $0< \varepsilon<1$ is some constant and  $\delta=N$ because we will have $N$ hashes of substrings of the string $s$. Additionally, we will use a hash function $h_p$, that is discussed in Section \ref{sec:roll-hash} 

Assume that the initial state for our algorithm is the following one
\begin{equation}\label{eq:state}
   \ket{\varphi}=\ket{h_p(w)}\otimes\bigotimes_{t=1}^{\log_2 N}\frac{1}{\sqrt{N}}\sum\limits_{a=0}^{N-1}\ket{a}
    \otimes \ket{h(s^{a+1})}.
\end{equation}

\subsubsection{Unique Target Case}

Firstly, assume that there is only one position $d$ such that $s^d=w$. In that case, we use only the following part of the quantum register.    
\[
   \ket{\varphi'}=\ket{h_p(w)}\otimes\left(\frac{1}{\sqrt{N}}\sum\limits_{i=0}^{N-1}\ket{i}
    \otimes \ket{h(s^{i+1})}\right).\]

Let us have a function $f:\{0,\dots,N-1\}\to\{0,1\}$, such that $f(i)=1$ iff $h_p(w)=h(s^{i+1})$. We will discuss the implementation of the algorithm for $f$ later.

Then, we can add additional qubit $\ket{\xi}$ in a state $\frac{1}{\sqrt{2}}(\ket{0}-\ket{1})$ and apply $O(\sqrt{N})$ times $D$ and $Q$ matrices to.$\ket{\varphi'}$ state. Here $D$ is the Grover's diffusion and $Q$ is the transformation that works in the following way:
\[Q:\ket{i}\ket{h(s^{i+1})}\ket{h_p(w)}\ket{\xi}\to\ket{i}\ket{h(s^{i+1})}\ket{h_p(w)}\ket{\xi \oplus f(i)}=\]\[= (-1)^{f(i)}\ket{i}\ket{h(s^{i+1})}\ket{h_p(w)}\ket{\xi}\]

So, using the idea of Grover's search algorithm, we can do $O(\sqrt{N})$ iterations of $D$ and $Q$ and after that measure the quantum register and obtain the index $d$ such that $h_p(w)=h(s^{d+1})$.

Let us discuss, the implementation of the function $f$. The computation of $f(i)$ is equivalent to the problem of checking equality  of two strings $z=h_p(w)$ and $z'=h(s^{i+1})$ that are stored in quantum memory. Let us define a function $g_i:\{0,\dots,\lceil\log_2 p\rceil-1\}\to\{0,1\}$, where $g_i(j)=1$ iff $z_j\neq z'_j$. In other words, we mark the indexes of unequal symbols of two strings.
We can say that $f(i)=0$ iff there is $j\in\{0,\dots,\log_2 p\}$ such that $g_i(j)=1$.

Let us solve this problem using Grover's search algorithm. In fact, we have two strings in unchangeable memory, and using additional $O(\log \log p)$ qubits can implement Grover's search algorithm for searching an index $j_1$ such that $g_i(j_1)=1$.

If the Grover's search algorithm finds $j_1$ and $g_i(j_1)=1$, then $f(i)=0$. If the result index  $j_1$ is such that $g_i(j_1)=0$, then $f(i)=1$.

Note that for a standard version of  Grover's search algorithm, function $f$ should be computed with no error. At the same time, our version of the implementation of $f$ can return a result with constant error probability. That is why we should use the modification of Grover's search algorithm \cite{hmw2003} for bounded-error oracle. This algorithm uses the generalization of Grover's search algorithm that is called Amplitude Amplification \cite{bhmt2002}.

\begin{lemma}\label{lm:sm-unique-target}
The presented algorithm solves string matching problem for unique target with bounded error, has $O(\sqrt{n}(\log n + \sqrt{\log n+\log m}\cdot (\log\log n+\log\log m)))$ time complexity and uses $O(\log n+\log m)$ qubits of memory.
\end{lemma}
\begin{proof}
Due to description of the algorithm, it finds the index $i$ such that $f(i)=1$, i.e. $h_p(w)=h(s^i)$ with constant probability. Let us say that the probability of success is at least $0.5$. Due to choice of $p$ and results that discussed in Section \ref{sec:roll-hash}, The fact $h_p(w)=h(s^i)$ means $w=s^i$ with probability at least $1-\varepsilon$. So the total probability of success is $0.5\cdot (1-\varepsilon)$. If we want a bigger success probability, then we can repeat the process several times and choose the major result. A similar technique was used, for example, in \cite{abikkpssv2020,kkmsy2020}. Constant times repetitions increase the total time complexity and memory size only in constant times.

Let us discuss the time complexity of the algorithm. Due to \cite{grover1996fast}, the time complexity of Grover's algorithm is $O(\sqrt{M}\log M)$  in a case of searching the target element among $M$ elements and constant time implementation of the oracle. In fact, the time complexity of all Grover's diffusion operators is  $O(\sqrt{M}\log M)$ and all Oracle operators is $O(\sqrt{M})$.  The modification with bounded-error oracle \cite{hmw2003} has only constant time bigger time complexity.  In our case $M=N$ and the oracle is complex. The implementation of $f$ function  has $O(\sqrt{\log p}\log \log p)$ time complexity because there are $\log_2 p$ objects for search. Hence, the total time complexity is $O(\sqrt{N}\log N + \sqrt{N}\cdot \sqrt{\log p}\log \log p)$. Here $p=\frac{\delta \cdot m}{\varepsilon}=\frac{N \cdot m}{\varepsilon}$ and $ O(\sqrt{\log p}\log \log p)= O((\log N+\log m -\log \varepsilon)^{0.5}\cdot(\log (\log N+\log m -\log \varepsilon))=O((\log N+\log m)^{0.5}\cdot (\log\log N+\log\log m))$. Therefor, the time complexity is
\[ O(\sqrt{N}\log N + \sqrt{N\cdot (\log N+\log m)}\cdot \sqrt{\log p}\log \log p)=\]\[=O(\sqrt{N}\log N + \sqrt{N(\log N+\log m)}\cdot (\log\log N+\log\log m))=\]
remember that $N=n-m$, therefore
\[=O(\sqrt{n}(\log n + \sqrt{\log n+\log m}\cdot (\log\log n+\log\log m))).\]

Let us discuss the memory complexity. The main part of the algorithm requires $O(\log N + \log p)$ qubits. Additionally, we need $O(\log\log p)$ for Grover's Search that implements the function $f$. Therefore, the total complexity is  
\[O(\log N + \log p+\log\log p)=O(\log N + \log N+\log m +\log\log N+\log\log m)=\]\[=O(\log N+\log m)=O(\log n+\log m).\]
Finally, we have proved the claim.
\end{proof}

\subsubsection{Multi-Target Case}
Let us consider the general case when the string $w$ can occur in $s$ several times. As mentioned in Section \ref{sec:grover}, we should repeat our algorithm $\log_2 N$ times with a different number of iterations. For several repetitions of the algorithm, we should have an unchangeable part of a quantum memory that holds all hashes $h_p(s^i)$. At the same time, our algorithm destroys the quantum state that holds the required data. 

Therefore, we should have $\log_2 N$ copies of our state that allow us to repeat the process several times. That is why we use the initial state in the (\ref{eq:state}) form.

Let us analyze the complexity of the algorithm.

\begin{theorem}\label{th:sm}
The presented algorithm solves string matching problem with bounded error, has $O(\sqrt{n}(\log n + \sqrt{\log n+\log m}\cdot (\log\log n+\log\log m)))$ time complexity and uses $O((\log n)^2+\log n \cdot \log m)$ qubits of memory.
\end{theorem}
\begin{proof}
Due to Lemma \ref{lm:sm-unique-target}, the algorithm for unique target solves the problem with bounded error, has $O(\sqrt{n}(\log n + \sqrt{\log n+\log m}\cdot (\log\log n+\log\log m)))$ time complexity and uses $O(\log n+\log m)$ qubits of memory.

Let us discuss time complexity. Due to \cite{bbht98,grover1996fast}, if we do $b$ iterations of the Grover's search algorithm for $M$ elements, then time complexity is $O(b\log M)$. In our algorithm we do $2^j$ iterations for $j$-th step, where $j\in\{0,\dots\lceil\log_2\sqrt{M}\rceil\}$ and $M=N$. Therefore, similar to the proof of Lemma \ref{lm:sm-unique-target}, we can show that  time complexity is 
\[O\left(\sum_{j=1}^{\log_2\sqrt{N})}2^j\left(\log n + \sqrt{\log n+\log m}\cdot (\log\log n+\log\log m)\right)\right)=\]
Due to the properties of the sum of geometric progression and $N=n-m$, we have
\[=O\left(\sqrt{n}\left(\log n + \sqrt{\log n+\log m}\cdot (\log\log n+\log\log m)\right)\right).\]

We have $\log_2{N}$ copies of qubits. Each of them is used for a single invocation of the algorithm for a unique target. Therefore, the total memory complexity is 
$O(\log n(\log n+\log m))=O((\log n)^2+\log n \cdot \log m)$.
\end{proof}

\section{Quantum Algorithm for String Comparing Problem}\label{sec:comp}
Let us discuss the algorithm for String Comparing Problem. There are two algorithms. The first one is based on Grover's search algorithm that was discussed in Section \ref{sec:grover}. The second one is faster and based on comparing strings using Binary search algorithm and rolling hash that was discussed in Section \ref{sec:roll-hash}, but it requires a more complex initial state.

\subsection{The Algorithm Based on Grover's Search Algorithm}

Let $k=\min(|u|,|v|)$ for strings $u$ and $v$. As it was discussed in Section \ref{sec:roll-hash}, for comparing two string $u$ and $v$, it is enough to find the Longest common prefix. We can use an idea similar to \cite{ki2019,kkmsy2020}. 
Let us consider a function $g':\{1,\dots,k\}\to\{0,1\}$ such that $k=\min(|u|,|v|)$, $g'(i)=1$ iff $u_i\neq v_i$. If we found the  smallest lexicographical element of the sequence $(1-g'(i),i)$ for $i\in\{1,\dots,k\}$, then it corresponds to the minimal argument $i_1$ such that $g'(i_1)=1$. Such idea is used in \cite{k2014,ll2015,ll2016,kkmsy2020} algorithms for searching the first target object.

We can use the D{\"u}rr-H{\o}yer algorithm for minimum search \cite{dh96,dhhm2006}. Let us briefly present its idea in Section \ref{sec:min} and then present algorithm itself in Section \ref{sec:comp1}.

\subsubsection{D{\"u}rr-H{\o}yer Minimum Search Algorithm}\label{sec:min}
The problem is searching for the index of minimal element among $(a_1,\dots,a_M\}$ for some positive integer $M$.

The algorithm contains several phases. The $0$-th phase is an assumption that minimal element $y^0$ is any element. On $i$-th phase, we have an assumption that the minimal element is $y^i$. Then, we run the Grover's search algorithm for searching for a smaller than $y^i$ element. We consider a function $g^{i}:\{1,\dots,M\}\to\{0,1\}$ such that $g^i(j)=1$ iff $a_j<y^i$. The algorithm finds any argument $j$ such that $g^i(j)=1$ and updates the assumption of the minimum by assigning  $y^{i+1}\gets a_j$, where $g^i(j)=1$.

Due to \cite{dhhm2006}, the expected number of phases is $O(\log M)$. At the same time, the expected number of all iterations of all invocations of Grover's search algorithm is $O(\sqrt{M})$. Due to Markov's inequality, if the algorithm stops after $3$ times more phases than the expectation, then we get a result with bounded error.  

Note that used Grover's search implementation is the algorithm for multi-target case.
\subsubsection{The Main Part of the Algorithm}\label{sec:comp1}
Assume that the initial state for our algorithm is the following one
\begin{equation}\label{eq:comp-main-state}
   \ket{\varphi}=\ket{\xi}\bigotimes_{i=1}^{3\log_2 k}\bigotimes_{t=1}^{\log_2 k}\frac{1}{\sqrt{k}}\sum\limits_{a=0}^{k-1}\ket{a}
    \otimes \ket{u_a}\otimes\ket{v_a}, \quad \ket{\xi}=\frac{1}{\sqrt{k}}\sum\limits_{a=0}^{k-1}\ket{a}
    \otimes \ket{u_a}\otimes\ket{v_a}.
\end{equation}



As in Section \ref{sec:smatch}, we implement Grover's search algorithm on our state. Let us discuss $i$-th phase of the algorithm.
We use 
\begin{equation}\label{eq:comp-phase-state}
\bigotimes_{t=1}^{\log_2 k}\frac{1}{\sqrt{k}}\sum\limits_{a=0}^{k-1}\ket{a}
    \otimes \ket{u_a}\otimes\ket{v_a}.
\end{equation}
Let $i=0$. We invoke Grover's search algorithm on the quantum state and find any $j_1$ such that $g(j_1)=1$. Note, that computing $g$ have constant time and memory complexity because it is comparing two qubits for equality. Then, we store $1-g(j_1)$  to a qubit $\ket{\phi^0}$ and we denote the obtained index as a qubit $\ket{\psi^0}$.

Let us consider the case of $i>0$.
Assume that we have a function $comp:\{0,1\}\times\{1,\dots,k\} \times \{0,1\}\times\{1,\dots,k\} \to \{0,1\}$ that compares two pairs $(q,i)$ and $(q',i')$ in lexicographical order, i.e. $comp(q,i,q',i')=1$ iff $q<q'$ or $(q=q') \& (i<i')$. The function can be implemented in constant time and memory complexity because each value is a single qubit.
We can say that $g^i(j)=comp(\ket{\phi^{i}},\ket{\psi^i},\ket{1-g'(j)}\ket{j})$.
Using the state (\ref{eq:comp-phase-state}) and $\ket{\phi^{i}}\ket{\psi^i}$ we can implement multi-target Grover's search as in Section \ref{sec:smatch}. After measurement we obtain a result index $j$ and value of the function $g(j)$. Then, we store them to the register $\ket{\phi^{i+1}}\ket{\psi^{i+1}}$.

Then, we do $3\log_2 k$ phases using new copies of the state (\ref{eq:comp-phase-state}). Finally, we obtain the minimal index $i_0$ of unequal symbols. We can access to $i_0$-th element of state $\ket{\xi}$, compare $u_{i_0}$ and $v_{i_0}$, and return the answer according to the discussion in Section \ref{sec:roll-hash}. For accessing to $i$-th element, we can swap it with $0$-th element using CNOT gates and then apply Hadamard transformation for collecting whole amplitude in $0$-th element. These operations require $O(\log k)$ time complexity.

\begin{theorem}\label{th:comp1}
The presented algorithm solves string comparing problem with bounded error. It has $O(\sqrt{k}\log k)$ time complexity and uses $O((\log k)^3)$ qubits of memory.
\end{theorem}
\begin{proof}
The algorithm solves the problem because it implements the idea from \cite{ki2019}. The probability of success is constant because of the properties of the D{\" u}rr-H{\o}yer algorithm for minimum search \cite{dh96,dhhm2006}.

Let us compute time complexity of the algorithm. Due to the properties of the D{\" u}rr-H{\o}yer algorithm, the total number of iterations of all invocations of Grover's search is $O(\sqrt{k}
)$. Time complexity of computing $g(i)$ and $comp$ are constant. Therefore, the total time complexity is $O(\sqrt{k}\log k)$.

Let us consider the memory complexity. We need $O((\log k)^3)$ qubits for state (\ref{eq:comp-main-state}).
\end{proof}

\subsection{The Algorithm Based on Binary Search}\label{sec:comp2}
Let us implement the idea with the Binary search algorithm that was discussed in Section \ref{sec:roll-hash-bins}.

Let $k=\min(|u|,|v|)$ for two comparing strings $u$ and $v$.
Let us choose a prime $p$ from the first $\frac{\delta \cdot k}{\varepsilon}$ prime numbers, where $0< \varepsilon<1$ is some constant and  $\delta=k$ because we will have $k$ hashes of substrings of the string $u$ and $v$.

Assume that the initial state for our algorithm is the following.
\begin{equation}\label{eq:comp2-phase-state}
\ket{\phi}\otimes\bigotimes_{t=1}^{\log_2 k}\frac{1}{\sqrt{k}}\sum\limits_{a=0}^{k-1}\ket{a}
    \otimes \ket{h(u[1,a+1])}\otimes\ket{h(v[1,a+1])}, \quad\ket{\phi}=\sum\limits_{a=0}^{k-1}\ket{a}\ket{u_a}
\end{equation}

We can find the first $a_0$ such that $h(u[1,a_0+1])\neq h(v[1,a_0+1])$ using Binary search algorithm because of arguments from Section \ref{sec:roll-hash-bins}.
On each phase, we should access to some middle element with an index $mid$ and compare $h(u[1,mid+1])$ and $h(v[1,mid+1])$. For accessing to $i$-th element we can swap it with $0$-th element using CNOT gates, and then apply Hadamard transformation for collecting whole amplitude in $0$-th element. These operations require $O(\log k)$ time complexity.

Next, we should compare two strings of length $O(\log p)$ for equality. We can do it using Grover's search algorithm as it was done in Section \ref{sec:smatch}. The time complexity of this algorithm is $O(\sqrt{\log p}\log \log p)$ and memory complexity is $O(\log\log p)$ qubits.

Therefore, after $\log_2 k$ steps of the Binary search algorithm, we obtain the minimal index $a_0$ such that  $h(u[1,a_0+1])\neq h(v[1,a_0+1])$. If  $h(u[1,a_0+1])= h(v[1,a_0+1])$, then $u=v$. If  $h(u[1,a_0+1])\neq h(v[1,a_0+1])$, then we can access to $a_0$-th element of $\ket{\phi}$ for accessing to $u_{a_0}$. If $u_{a_0}=0$, then $u<v$ and $u>v$ otherwise. 

\begin{theorem}\label{th:comp2}
The presented algorithm solves string comparing problem with bounded error. It has $O((\log k)^{2}\log \log k)$ time complexity and uses $O((\log k)^2)$ qubits of memory.
\end{theorem}
\begin{proof}
The algorithm solves the problem because it implements the idea from Section \ref{sec:roll-hash-bins}.

Let us compute time complexity of the algorithm. There are $O(\log k)$ phases of Binary search. Each phase requires comparing hashes in $O(\sqrt{\log p} \log \log p)$ and accessing to $mid$-th element in $O(\log k)$.
The final step is accessing to element with $O(\log k)$ time complexity. The final time complexity is

$O(\log k \cdot (\sqrt{\log p} \log \log p + \log k) + \log k) =$ $O((\log k)\cdot (\sqrt{\log k} \log\log k+\log k))=$ $O((\log k)^{2}\log \log k)$.

Let us consider the memory complexity. We need $O(\log k\cdot(\log k+ \log p) + \log k)=O((\log k)^2)$ qubits for state (\ref{eq:comp2-phase-state}) and $O(\log\log p)=O(\log\log k)$ states for the implementation of two hashes comparing. So, the total memory complexity is $O((\log k)^2)$.
\end{proof}

\section{Conclusion}\label{sec:conclusion}
In the paper, we presented algorithms for two problems - String matching problem and String comparing problem. The algorithm for the String matching problem works as fast as the best-known quantum algorithm up to a log factor. At the same time, our algorithm uses exponentially fewer qubits of memory. We have presented two algorithms for string comparing problem. Both use exponentially fewer qubits comparing to the best-known algorithm for the problem. The first one is based on Grover's search algorithm and uses more qubits than the second one based on Binary search. The second algorithm works exponentially faster than the first one and than the existing algorithm \cite{ki2019}. At the same time, the initial state of the second algorithm is more complex compared to the initial state of the first algorithm. 

The initial state of all algorithms is not just stored input in quantum memory. At the same time, preparing this state is not much harder than storing input data in quantum memory as is. Additionally, these algorithms can be used as a part of other algorithms. A similar motivation is presented in different papers, for example, in \cite{hhl2009}.

\begin{acknowledgement}
The research is funded by the subsidy allocated to Kazan Federal University
for the state assignment in the sphere of scientific activities, project No. 0671-2020-0065.
A part of the reported study (Section \ref{sec:comp}) was funded by RFBR according to the research project No.20-37-70080. 
\end{acknowledgement}


%


\end{document}